\documentclass[12pt]{article}

\usepackage[left=2cm,top=2.5cm,right=2cm,bottom=3cm,nohead]{geometry}
\usepackage{amsmath,amsfonts,amsthm,amssymb, bm}
\usepackage{epsfig}
\usepackage{verbatim,epic, eepic}
\input amssym.def
\input amssym.tex
\usepackage{yfonts}[1998/10/03]

\newfont{\go}{ygoth.tfm scaled 1200}  

\newtheorem{theorem}{Theorem}[section]
\newtheorem{lemma}[theorem]{Lemma}

\newtheorem{corollory}[theorem]{Corollary}

\newcommand{\norm}[1]{\left\|#1\right\|}

\newcommand{\com}[2]{ \left [ {#1}, {#2} \right ]}
   
\newcommand{\cwedge}[1]{\mathop{\wedge}_{{}^{#1}} }
\newcommand{\hook}{\raisebox{-0.35ex}{\makebox[0.6em][r]
{\scriptsize $-$}}\hspace{-0.15em}\raisebox{0.25ex}{\makebox[0.4em][l]{\tiny
 $|$}}}
 \newcommand{\HRule}{\rule{\linewidth}{0.25mm}}

\numberwithin{equation}{section}
\newcommand{\n}[1]{\label{#1}}
\newcommand{\be}{\begin{equation}}
\newcommand{\ee}{\end{equation}}
\newcommand{\ba}{\begin{eqnarray}}
\newcommand{\ea}{\end{eqnarray}}

\newcommand{\beq}{\begin{equation}}
\newcommand{\eeq}{\end{equation}}
\newcommand{\beqa}{\begin{eqnarray}}
\newcommand{\eeqa}{\end{eqnarray}}

\newcommand{\eq}[1]{(\ref{#1})}

\newcommand{\nt}[1]{\nabla^T_{X_{#1}}}
\newcommand{\dt}{d^T\!}
\newcommand{\delt}{\delta^T\!}
\newcommand{\mD}{{\cal D}}

\begin{document}
\begin{flushright}
\small
DAMTP-2010-14 \\
OCU-PHYS 326 \\
\vspace{.5cm}
\today \\
\normalsize
\end{flushright}

\begin{center}

\vspace{1.5cm}

{\LARGE {\bf Symmetries of the Dirac operator with skew-symmetric torsion \\}}

\vspace{1cm}

{Tsuyoshi Houri\let\thefootnote\relax\footnotetext{{\em Emails}: houri@sci.osaka-cu.ac.jp (T. Houri), D.Kubiznak@damtp.cam.ac.uk (D. Kubiz\v n\'ak), \\ \mbox{} \hspace{1.8cm}c.m.warnick@damtp.cam.ac.uk (C. Warnick), yasui@sci.osaka-cu.ac.jp (Y. Yasui)}$^{, a}$, David Kubiz\v n\'ak$^{2, b}$, Claude Warnick$^{3, b, c}$, Yukinori Yasui$^{4, a}$}

\end{center}

{\footnotesize $^a$ Department of Mathematics and Physics, Graduate School of Science, \\ \mbox{}\hspace{.8cm} Osaka City University, 3-3-138 Sugimoto, Sumiyoshi, Osaka 558-8585, JAPAN}

{\footnotesize $^b$ DAMTP, University of Cambridge, Wilberforce Road, Cambridge CB3 0WA, UK}

{\footnotesize $^c$ Queens' College, Cambridge, CB3 9ET}
\\
\\
\noindent
\HRule

\begin{abstract}
In this paper, we consider the symmetries of the Dirac operator
derived from a connection with skew-symmetric torsion, $\nabla^T$. We find that
the generalized conformal Killing--Yano tensors give rise to symmetry
operators of the massless Dirac equation, provided an explicitly given
anomaly vanishes. We show that this gives rise to symmetries of the Dirac operator in the case of strong K\"ahler with torsion (KT) and strong hyper-K\"ahler with torsion (HKT) manifolds.
\end{abstract}
\vspace{-.1cm}
\noindent
\HRule
\section{Introduction}

In the study of symmetries of pseudo-Riemannian manifolds, an important
role is played by conformal Killing--Yano tensors 
\cite{Yano, Yano:1953}. These generalize
conformal Killing vectors to higher rank anti-symmetric tensors and have many elegant properties which can be exploited 
both in mathematical and physical contexts. We refer the reader to four recent PhD dissertations and the references therein
\cite{Kress}.

Recently, there has been interest in the generalization of the
conformal Killing--Yano tensors to pseudo-Riemannian manifolds with
skew-symmetric torsion first introduced in \cite{Yano:1953}, then
re-discovered in \cite{Kubiznak:2009qi,  Wu:2009cn}. Such manifolds
occur in
superstring theories where the torsion form may be identified (up to a
factor) with a $3$-form field occurring naturally in the
theory. Geometries with torsion are also of interest in the context of
non-integrable special Riemannian geometries~\cite{Agricola:2006tx}. Important examples include K\"ahler
 with torsion (KT) and hyper K\"ahler with torsion (HKT) manifolds
 \cite{Howe:1996kj, Grantcharov:1999kv}, which have applications across mathematics and
 physics.

A key property of conformal Killing--Yano tensors in the absence of
torsion is that one may construct from them symmetry operators for the
massless Dirac equation \cite{Carter:1979fe, Benn:1996ia, Cariglia:2003kf}. Conversely any first-order symmetry operator
of the massless Dirac equation is constructed from a conformal
Killing--Yano tensor \cite{Benn:2004, Acik:2008wv}. In
this paper, we will consider the question of when a {\em generalized conformal Killing--Yano} (GCKY) 
tensor with respect to a connection with torsion $T$ gives rise to a symmetry of an appropriate massless Dirac
equation with torsion. For reasons which we will motivate below, the
appropriate Dirac operator to consider is not that constructed from
the connection on the spin bundle with torsion $T$, which we denote $\nabla^T$, but that constructed from the
connection with torsion $T/3$, which we denote $\mD$ and can relate to the Dirac operator of the connection $\nabla^T$ by
\be
\mD = D^T + \frac{T}{2}\,.
\ee
This `modified' Dirac operator is that introduced by Bismut in \cite{Bismut:1989}.
One may alternatively choose to view our result as relating the symmetries of
the Dirac operator with torsion $T$ to the GCKY tensors of the
connection with torsion $3T$.

We now state the main theorem of the paper, see subsequent sections
for conventions.
\begin{theorem} \label{mainthm}
Let $\omega$ be a generalized conformal Killing-Yano (GCKY) $p$--form obeying
\begin{equation}
\nabla^T_X \omega-\frac{1}{p+1}{X}\hook {d}^T \omega+\frac{1}{n-p+1}
{X}^\flat \wedge {\delta}^T \omega=0,
\end{equation}
then the operator
\be
L_\omega=e^a\omega\nt{a}+\frac{p}{p+1}\dt \omega -\frac{n-p}{n-p+1}\delt\omega+\frac{1}{2}T\omega
\ee
satisfies
\be
\mD L_\omega=\omega \mD^2+\frac{(-1)^p}{p+1}\dt \omega \mD+\frac{(-1)^p}{n-p+1}\delt\omega\mD-A\,.
\ee
The anomaly $A$, given explicitly below, contains no term where a
derivative acts on a spinor and depends on $T$ and $\omega$. In the case where $A$ vanishes,
$L_\omega$ is a symmetry operator for the massless Dirac equation.
\end{theorem}
We also exhibit this operator in terms of $\gamma$-matrices, a
formalism perhaps more familiar to physicists. In the cases where
either $d^T\omega=0$ or $\delta^T \omega =0$, we
can exhibit operators which (anti)-commute with the Dirac operator.
\begin{corollory} \label{GKYcor}
Let $\omega$ be a generalized Killing-Yano (GKY) $p$-form such that
$A$ vanishes. Then there exists an operator $K_\omega$ such that
\be
\mD K_\omega+(-1)^p K_\omega \mD = 0.
\ee
\end{corollory}
\begin{corollory} \label{GCCKYcor}
Let $\omega$ be a generalized closed conformal Killing-Yano (GCCKY) $p$-form such that
$A$ vanishes. Then there exists an operator $M_\omega$ such that
\be
\mD M_\omega-(-1)^p M_\omega \mD = 0.
\ee
\end{corollory}

The anomaly term of Theorem \ref{mainthm} obviously vanishes for
vanishing torsion, by the result of Benn and Charlton \cite{Benn:1996ia}. It is also
shown to vanish in \cite{Houri:2010} for the case of the GCKY tensors
exhibited by the charged Kerr-NUT metrics \cite{Chow:2008fe} (including the Kerr-Sen
black hole spacetime \cite{Sen}). The anomaly also vanishes in the case of strong
K\"ahler with torsion (KT) and strong hyper-K\"ahler with torsion (HKT) metrics, giving the result:
\begin{corollory} \label{HKTcor}
A strong KT metric admits one operator and a strong HKT metric three operators which commute with the
modified Dirac operator $\mD$.
\end{corollory}

In section \ref{prelim} we introduce our notation and conventions and derive some
basic results which we require in the proof of Theorem
\ref{mainthm}. Section \ref{CCKYsec} introduces the GCKY tensors,
together with some of their properties. The proof of the main theorem
takes up section \ref{proof}. 

\section{Preliminaries \label{prelim}}

We assume that we work on $(M^n, g)$, a pseudo-Riemannian spin manifold. It
is convenient to denote by $\{ X_a \}$ a local orthonormal basis for
$TM$ and by $\{ e^a \}$ the dual basis for $T^* M$. We also define
\begin{equation}
\eta_{ab} = g(X_a, X_b), \qquad \eta^{ab} =
\left(\eta^{-1}\right)^{ab}, \qquad X^a = \eta^{ab} X_b, \qquad e_a =
\eta_{ab} e^b.
\end{equation}
The matrix with entries $\eta_{ab}$ will of course be diagonal, with
entries $\pm 1$. Throughout we sum over repeated indices.

\subsection{The Clifford algebra} \label{Cliffsec}

We will work with the conventions of Benn and Tucker \cite{BennT}, to
whom we refer the reader for more details of this formalism. We
recall that differential forms may be identified with elements of the
Clifford algebra. Denoting the Clifford product by juxtaposition, our
convention\footnote{note that in the mathematical literature, a different
  convention with the signs of the interior derivatives reversed is
  often used} is that for a $1$-form $\alpha$ and \ $p$-form $\omega$
\begin{eqnarray}
\alpha \omega &=& \alpha \wedge \omega + \alpha^\sharp \hook \omega\,,
\nonumber \\
\omega \alpha &=& (-1)^p \left(\alpha \wedge \omega - \alpha^\sharp
  \hook \omega \right)\,. \label{Cliff}
\end{eqnarray}
Repeated application of this rule allows us to construct the Clifford
product between forms of arbitrary degree. Appendix
\ref{cliffprods}  contains, for convenience, expressions for the products and
some (anti)-commutators. In order to state these products compactly, we introduce the
contracted wedge product, defined inductively by
\begin{equation}
	\alpha \cwedge{0} \beta = \alpha \wedge \beta, \qquad \alpha \cwedge{n} \beta = \left(X_a \hook \alpha\right)\mathop{\wedge}_{{}^{n-1}} \left( X^a \hook \beta \right).
\end{equation}

For those more familiar with the Clifford algebra in terms of
$\gamma$-matrices, we note that the identification between
differential forms and elements of the Clifford algebra can be
expressed as
\begin{equation}
\frac{1}{p!}\omega_{a_1 \ldots a_p} e^{a_1} \wedge \cdots \wedge
e^{a_p} \longrightarrow \frac{1}{p!}\omega_{a_1 \ldots a_p}
\gamma^{[a_1} \cdots \gamma^{a_p]}. \label{gammacor}
\end{equation}
The relations (\ref{Cliff}) are equivalent to the usual
anti-commutator for the $\gamma$-matrices:
\begin{equation}
\gamma^a \gamma^b + \gamma^b \gamma^a = 2 \eta^{ab}.
\end{equation}
We will use the short-hand
$e^{a_1 \ldots a_p} = e^{[a_1} \cdots e^{a_p]} = e^{a_1} \wedge \cdots \wedge
e^{a_p}.$

\subsection{The connection with torsion}

We wish to consider a connection with totally anti-symmetric torsion
$T \in \Omega^3 (M)$ which is defined, for arbitrary vectors $X$ and $Y$, in terms of the Levi-Civita
connection $\nabla$ to be
\begin{equation}
\nabla^T_X Y = \nabla_X Y + \frac{1}{2} T(X, Y, X_a) X^a\,. \label{conn}
\end{equation}
Acting on a form, we find that
\begin{equation}
\nabla^T_X \omega = \nabla_X \omega + \frac{1}{2} (X \hook T)
\cwedge{1} \omega\,.
\end{equation}
This connection is metric, so obeys
$X(g(Y, Z)) = g(\nabla^T_X Y, Z) + g(Y, \nabla^T_X Z)$
and in addition has the same geodesics as the Levi-Civita connection
since
$\nabla^T_X X = \nabla_X X\,.$

It is useful to define two operators on forms related to the exterior
derivative and its adjoint as
\begin{eqnarray}
d^T \omega &=& e^a \wedge \nabla^T_{X_a} \omega = d \omega - T
\cwedge{1} \omega\,, \nonumber \\
\delta^T \omega &=& - X^a \hook \nabla^T_{X_a} = \delta \omega -
\frac{1}{2}  T \cwedge{2} \omega. \label{extder}
\end{eqnarray}
These respectively raise and lower the degree of a form by one. Unlike
the standard $d, \delta$ they are not nilpotent in general.

In calculations it is often convenient to consider a basis which is
parallel at a point $p$. For a connection with non-zero torsion, such
a basis does not necessarily exist. It is always possible to arrange that at $p$:
\begin{eqnarray}
\nabla_{X_a}^T X_b &=& \frac{1}{2}T(X_a, X_b, X_c) X^c\,, \nonumber\\
\nabla_{X_a}^T e^b &=& \frac{1}{2}T(X_a, X^b, X_c) e^c\,, \nonumber\\
 \com{X_a}{ X_b} &=& 0\,. \label{basiscond}
\end{eqnarray}
We will assume from now on that we work with such a basis.

\subsubsection{Curvature}

The curvature operator is defined as usual by
\begin{equation}
R(X,Y) \omega = \left( \nabla^T_X \nabla^T_Y -\nabla^T_Y \nabla^T_X -
  \nabla^T_{\left[ X, Y \right]}  \right) \omega.
\end{equation}
In Appendix \ref{bianchisec} we prove versions of the first Bianchi
identity for the curvature operator. 

For any pair of vectors $X, Y$, we
introduce the following $2$-form, following \cite{BennT}
\begin{equation}
\mathcal{R}_{X,Y} = \frac{1}{4} ( X_a \hook R(X, Y) e_b)\, e^{ab} =
-\frac{1}{4} X \hook Y \hook R_{ab} e^{ab}.
\end{equation}
Here $R_{ab}$ are the usual curvature $2$-forms. The reason for
introducing this $2$-form is as follows. Using the results of Appendix
\ref{cliffprods} it can be readily demonstrated that
$\left[\mathcal{R}_{X,Y}, e^a \right] = R(X, Y) e^a,$
where the commutator is taken in the Clifford sense. One may extend
this result by linearity to show that for any form $\omega$, the
curvature operator is related to the $2$-form $\mathcal{R}_{X,Y}$ by
\begin{equation}
\left[\mathcal{R}_{X,Y}, \omega \right] = R(X, Y) \omega.
\end{equation}
In Appendix \ref{bianchisec} we establish the following Bianchi identity for
$\mathcal{R}_{X,Y}$:
\begin{equation}
e^{ab} \mathcal{R}_{X_a,X_b} = -\frac{3}{2} d^TT - T \cwedge{1} T +
\delta^T T - \frac{1}{2} s\,. \label{bianchi2}
\end{equation}
Here $s$ is the scalar curvature of the connection with torsion, which
we take to be defined by
$s = - X^a \hook R(X_a, X_b) e^b$.

\subsubsection{The lift to the spinor bundle}

We will consider the natural lift of the connection (\ref{conn}) to
the spinor bundle, given by
\begin{equation}
\nabla^T_X \psi = \nabla_X \psi - \frac{1}{4} X \hook T \psi,
\end{equation}
where $\psi$ is a spinor field (i.e.\ a section of the spinor bundle)
and $\nabla_X$ is the usual lift of the Levi-Civita connection. This
connection is a derivation over the Clifford product between forms and
spinors, 
\begin{equation}
\nabla^T_X(\omega \psi) = (\nabla^T_X \omega) \psi + \omega \nabla^T_X \psi,
\end{equation}
for any form $\omega$ and spinor field $\psi$. It is shown in
\cite{BennT} that the curvature operator acting on a spinor is given
simply by
\begin{equation}
R(X,Y) \psi = \left( \nabla^T_X \nabla^T_Y -\nabla^T_Y \nabla^T_X -
  \nabla^T_{\left[ X, Y \right]}  \right) \psi = \mathcal{R}_{X,Y} \psi\,.
\end{equation}

\subsection{The Dirac operator}

The `na\"ive' or `bare' Dirac operator defined by the spinor
covariant derivative given above is
\begin{equation}
D^T = e^a \nabla^T_{X_a} = D - \frac{3}{4} T,
\end{equation}
where $D$ is the Dirac operator of the Levi-Civita connection. For the
purposes of the physics of spinning particles, among other reasons, it
is convenient to consider a modified operator of the form $D^T +
\alpha T$ which satisfies a Schr\"odinger-Lichnerowicz type
formula \cite{Agricola:2004}. 

Making use of the Clifford algebra, together with (\ref{bianchi2}) one may show that
\begin{equation}
(D^T)^2  = -X^a \hook T \nabla^T_{X_a} - \Delta^T + \frac{1}{2}
e^{ab}\mathcal{R}_{X_a, X_b},
\end{equation}
where we have introduced the spinor Laplacian
\begin{equation}
\Delta^T = - \nabla^T_{X_a} \nabla^T_{X^a} + \nabla^T_{\nabla^T_{X_a}
  X^a}.
\end{equation}
This suggests that a field obeying the massless Dirac equation will not
obey the natural spinor Klein-Gordon equation associated with the
connection $\nabla^T$ as there is an additional term linear in
derivatives of the field. In order to remove this unwanted term, we
note that
\begin{equation}
T D^T + D^T T = 2 X_{a}\hook T \nabla^T_{X_a} + d^TT - \delta^T T,
\end{equation}
so that
\begin{eqnarray}
-(D^T+ \alpha T)^2 &=& \Delta^T + (1-2\alpha) X_{a}\hook T
\nabla^T_{X_a} + (\alpha-\frac{1}{2})\delta^T T +
(\frac{3}{4}-\alpha)d^T T \nonumber \\ && + (\frac{1}{2}-\alpha^2)
T\cwedge{1} T + \frac{\alpha^2}{6}  \norm{T}^2 + \frac{s}{4}, \label{SchrLich}
\end{eqnarray}
where we have made use of the results of Appendix \ref{cliffprods} to
express the Clifford product of $T$ with itself. 

We may think of (\ref{SchrLich})
as a formula relating the square of a Dirac operator with torsion
$(1-4 \alpha/3) T$  to derivatives involving the connection $T$. Clearly the term
linear in derivatives will vanish if we take $\alpha = 1/2$. It is
therefore natural when we have a connection with torsion $T$ to
consider the modified Dirac operator \cite{Bismut:1989, Agricola:2006tx, Kassuba}
\begin{equation}
\mD \equiv D^T + \frac{T}{2}= D^{T/3}\,.
\end{equation}
In terms of the Levi-Civita connection and with the help of correspondence \eq{gammacor} this reads  
\be\n{Dirac}
\mD=D-\frac{1}{4}T=\gamma^a\nabla_a-\frac{1}{24}T_{abc}\gamma^{abc}\,.
\ee
Such an operator satisfies the elegant relation
\begin{equation}
\mD^2 =-\Delta^T - \frac{dT}{4} - \frac{s}{4} - \frac{\norm{T}^2}{24}\,.
\end{equation}
Up to quantum corrections this gives the Hamiltonian considered at the
spinning particle level in \cite{Rietdijk:1995ye, DeJonghe:1996fb, Kubiznak:2009qi}. It will be convenient for us to have the alternative form
\be
\mD^2=-\Delta^T+\frac{1}{2}e^{ab}{\cal R}_{X_a X_b}+\frac{1}{2}(\dt T-\delt T)+\frac{1}{4} T^2\,. \label{dirsq}
\ee

\section{Generalized conformal Killing--Yano tensors \label{CCKYsec}}

A $p$-form $\omega$ is called a {\em generalized conformal
  Killing-Yano} (GCKY) tensor \cite{Yano:1953, Kubiznak:2009qi} if it obeys the conformal Killing--Yano equation with torsion:
\beq\label{GCKY}
\nabla^T_X \omega-\frac{1}{p+1}{X}\hook {d}^T \omega+\frac{1}{n-p+1}
{X}^\flat \wedge {\delta}^T \omega=0
\eeq
for any $X\in TM$. If $\delta^T\omega=0$, $\omega$ is called a generalized Killing--Yano
(GKY) tensor, while if $d^T\omega=0$ it is a generalized closed conformal Killing--Yano  (GCCKY) tensor. Eq. \eq{GCKY} is conformally invariant and also invariant under $\omega\to*\omega$.  The Hodge dual interchanges GKY and GCCKY tensors.

A key property of GCCKY tensors is that they form a (graded) algebra under the
wedge product \cite{Kubiznak:2009qi, KrKuPaFr}. We summarize this in a lemma
\begin{lemma} \label{gcckylem}
If $\omega_1$ and $\omega_2$ are GCCKY tensors then so is $\omega_1 \wedge
\omega_2$. Further, the operator
\be
\delta^T \circ (n-\pi+1)^{-1},
\ee
where $\pi$ is the linear map taking a $p$-form $\alpha$ to $p \alpha$, is a graded
derivation on the exterior algebra of GCCKY tensors.
\end{lemma}
\begin{proof}
First note that $\omega\in \Omega^p M$ is GCCKY if and only if there exists an
$\tilde{\omega}$ such that $\nt{} \omega = X^\flat \wedge \tilde{\omega}$
and if this holds, $\tilde{\omega} = -(n-p+1)^{-1}
\delta^T\omega$. Suppose  $\omega_i$ are GCCKY forms of
degree $p_i$. We calculate
\begin{equation*}
\nt{}(\omega_1 \wedge \omega_2) = \nt{}\omega_1 \wedge
\omega_2+\omega_1 \wedge \nt{}\omega_2 = -X^\flat \wedge
\left(\frac{\delta^T \omega_1\wedge \omega_2}{n-p_1+1} + (-1)^{p_1} \frac{\omega_1\wedge\delta^T\omega_2}{n-p_2+1} \right).
\end{equation*}
\end{proof}

\subsection{Integrability conditions}
We will require integrability conditions for the GCKY equation. By
differentiating \eq{GCKY} we find the following conditions,
generalizing results of Semmelmann \cite{Semmelmann:2002}.
\ba
\Delta^T \omega &=&
\frac{1}{p+1}\delt\dt\omega+\frac{1}{n-p+1}\dt\delt\omega\,, \\
e^b\wedge X^a\hook
R(X_a,X_b)\omega&=&\frac{p}{p+1}\delt\dt\omega+\frac{n-p}{n-p+1}\dt\delt\omega+(X_b\hook
T)\cwedge{1}\nt{b}\omega\, . \label{intcond2}
\ea
Here $\Delta^T = -\nt{a}\nabla^T_{X^a}+\nabla^T_{\nt{a} X^a}$ is the Bochner Laplacian derived from the connection
$\nabla^T$. These are consistent with the Weitzenb\"ock identity in the
presence of torsion:
\be
\delt\dt+\dt\delt = \Delta^T+e^b\wedge X^a\hook
R(X_a,X_b)-(X_b\hook T)\cwedge{1}\nt{b}\, .
\ee
The condition \eq{intcond2} was already known to Yano and Bochner \cite{Yano:1953} in the case where $\delta^T \omega$ vanishes.

At this stage it will be convenient to deduce an identity relating the curvature form  
${\cal R}_{XY}$ to the integrability condition. 
We consider
\ba
\frac{1}{2}[e^{ab},\omega]{\cal R}_{X_a X_b}\!\!\!&-&\!\!\!\Bigl(\frac{p}{p+1}\delt\dt\omega +\frac{n-p}{n-p+1}\dt\delt\omega\Bigr)\nonumber\\
&=&
\frac{1}{2}[e^{ab},\omega]{\cal R}_{X_a X_b}-e^b\wedge X^a\hook R(X_a,X_b)\omega+(X_b\hook T)\cwedge{1}\nt{b}\omega\nonumber\\
&=&
\frac{1}{2}[e^{ab},\omega]{\cal R}_{X_a X_b}+\frac{1}{2}e^{ab}[{\cal R}_{X_a X_b},\omega]-\frac{1}{2} e^{ab}\wedge R(X_a,X_b)\omega\nonumber\\
&&\quad-\frac{1}{2}X^a\hook X^b\hook R(X_a,X_b)\omega+
(X_b\hook T)\cwedge{1}\nt{b}\omega\nonumber\\
&=&\frac{1}{2}[e^{ab}{\cal R}_{X_a X_b},\omega]-\dt\dt\omega-(X^a\hook T)\wedge \nt{a}\omega-\delt\delt\omega+\frac{1}{2}(X^a\hook T)\cwedge{2}\nt{a}\omega\nonumber\\
&&\quad+(X_b\hook T)\cwedge{1}\nt{b}\omega\,.\nonumber
\ea
Here we have made liberal use of the results of appendices \ref{cliffprods} and \ref{bianchisec}. Putting this together with \eq{bianchi2} we find 
\ba
\frac{1}{2}[e^{ab},\omega]{\cal R}_{X_a X_b}\!\!\!&-&\!\!\!\Bigl(\frac{p}{p+1}\delt\dt\omega +\frac{n-p}{n-p+1}\dt\delt\omega\Bigr)\nonumber\\
&=& \frac{1}{2}\Bigl[-\frac{3}{2}\dt T-T\cwedge{1}T+\delt T, \omega\Bigr]-\dt\dt\omega-\delt\delt\omega-(X^a\hook T)\nt{a}\omega\,. \label{intcond}
\ea

\section{Symmetry operators \label{proof}}
\subsection{Massless Dirac}
We are now ready to construct a symmetry operator for the Dirac operator. Our goal is to show that (up to anomaly terms) the operator
\be
L_\omega=e^a\omega\nt{a}+\frac{p}{p+1}\dt \omega -\frac{n-p}{n-p+1}\delt\omega+\frac{1}{2}T\omega \label{Ldef}
\ee 
satisfies
\be
\mD L_\omega=\omega \mD^2+\frac{(-1)^p}{p+1}\dt \omega \mD+\frac{(-1)^p}{n-p+1}\delt\omega\mD,
\ee
provided $\omega$ is a GCKY form and to calculate the anomaly terms. This means that when the anomaly terms (given explicitly by \eq{anomaly} below) vanish, the operator $L_\omega$ gives an {\em on-shell} ($\mD\psi=0$) symmetry operator for $\mD$. 

In the rest of this section, we shall sketch the proof of these assertions. It essentially consists of commuting $\mD$ through the operator $L_\omega$. We calculate 
\ba
\mD e^a\omega\nt{a}\!\!\!&=&\!\!\!-\omega\Delta^T+\frac{1}{2}e^{ab}\omega{\cal R}_{X_a X_b}+2\nt{a}\omega\nt{a}-e^b(\dt \omega-\delt\omega)\nt{b}-X^b\hook T\omega\nt{b}+\frac{1}{2}Te^b\omega\nt{b}\,,\nonumber\\
\mD \dt \omega \!\!\!&=&\!\!\!(\dt \dt \omega-\delt\dt\omega)+e^a\dt\omega\nt{a}+\frac{1}{2}T\dt \omega\,,\nonumber\\
\mD \delt \omega \!\!\!&=&\!\!\!(\dt\delt\omega-\delt\delt\omega)+e^a\delt\omega\nt{a}+\frac{1}{2}T\delt \omega\,,\nonumber\\
\mD T \omega \!\!\!&=&\!\!\! (\dt T-\delt T)\omega+e^a T\nt{a}\omega+e^aT\omega\nt{a}+\frac{1}{2}T^2\omega\,,
\ea
We need to simplify the sum of these four terms (with appropriate coefficients as given by \eq{Ldef}). Let us first gather the terms with one derivative acting on a spinor
\ba
\Bigl(2\nt{a}\omega\!\!&-&\!\!\frac{1}{p+1}e^a\dt \omega+\frac{1}{n-p+1}e^a\delt\omega\Bigr)\nt{a}+\bigl(-X^b\hook T+\frac{1}{2}T e^b+\frac{1}{2}e^bT\bigr)\omega \nt{b}\nonumber\\
\!\!&=&\!\!2\Bigl(\nt{a}\omega -\frac{X^a\hook \dt \omega}{p+1}+\frac{e^a\wedge\delt\omega}{n-p+1}\Bigr)\nt{a}+\frac{(-1)^p}{p+1} \dt \omega e^a\nt{a}+\frac{(-1)^p}{n-p+1}\delt\omega e^a\nt{a}\nonumber\\
\!\!&=&\!\!\frac{(-1)^p}{p+1}\dt \omega \mD+\frac{(-1)^p}{n-p+1}\delt\omega\mD-\frac{(-1)^p}{2(p+1)}\dt \omega T-\frac{(-1)^p}{2(n-p+1)}\delt \omega T\,,
\ea 
where in order to remove differential terms not proportional to $\mD$ we impose the GCKY equation \eq{GCKY}. This condition on $\omega$ arises also at the spinning particle level \cite{Kubiznak:2009qi}.
So we have
\ba
\mD L_\omega\!\!&=&\!\!
\Bigl(2\nt{a}\omega-\frac{1}{p+1}e^a\dt \omega+\frac{1}{n-p+1}e^a\delt\omega\Bigr)\nt{a}+\bigl(-X^b\hook T+\frac{1}{2}T e^b+\frac{1}{2}e^bT\bigr)\omega \nt{b}\nonumber\\
\!\!&&\!\!\quad-\omega\Delta^T+\frac{1}{2}e^{ab}\omega{\cal R}_{X_a X_b}+\frac{1}{2}(\dt T-\delt T)\omega+\frac{1}{4}T^2\omega\nonumber\\
\!\!&&\!\!\quad+\frac{1}{2}e^aT\nt{a}\omega+\frac{p}{2(p+1)}T\dt \omega-\frac{n-p}{2(n-p+1)}T\delt \omega\nonumber\\
\!\!&&\!\!\quad+\frac{p}{p+1}\dt \dt \omega+\frac{n-p}{n-p+1}\delt\delt\omega-
\Bigl(\frac{p}{p+1}\delt\dt\omega+\frac{n-p}{n-p+1}\dt\delt\omega\Bigr)\nonumber\\
\!\!&=&\!\!\frac{(-1)^p}{p+1}\dt \omega \mD+\frac{(-1)^p}{n-p+1}\delt\omega\mD-\frac{(-1)^p}{2(p+1)}\dt \omega T-\frac{(-1)^p}{2(n-p+1)}\delt \omega T\nonumber\\
\!\!&&\!\!\quad+\omega\mD^2+\frac{1}{2}[e^{ab},\omega] {\cal R}_{X_a X_b}+\frac{1}{2}\bigl[\dt T-\delt T+\frac{1}{2}T^2,\omega\bigr]\nonumber\\
\!\!&&\!\!\quad+X^a\hook T\nt{a}\omega-\frac{T\dt \omega}{2(p+1)}+\frac{T\delt \omega}{2(n-p+1)}\nonumber\\
\!\!&&\!\!\quad+\frac{p}{p+1}\dt \dt \omega+\frac{n-p}{n-p+1}\delt\delt\omega-
\Bigl(\frac{p}{p+1}\delt\dt\omega+\frac{n-p}{n-p+1}\dt\delt\omega\Bigr)\,,\nonumber
\ea
where we have used \eq{dirsq}. Consolidating, we obtain 
\be
\mD L_\omega=\omega\mD^2+\frac{(-1)^p}{p+1}\dt \omega \mD +\frac{(-1)^p}{n-p+1}\delt\omega\mD-A\,,
\ee
where we define $A$ to be the anomaly, i.e., the obstruction to $L_\omega$ giving a symmetry operator of the massless Dirac equation. This completes the proof of Theorem \ref{mainthm}.

In order to find a compact expression for the anomaly $A$, we make use of the integrability condition \eq{intcond} to find
\be\label{anomaly2}
A =
\frac{1}{2}\left(\frac{[T,\dt \omega]_p}{p+1}-\frac{[T,\delt \omega]_{p+1}}{n-p+1}\right)
+\frac{1}{4}[ dT,\omega]+\left(\frac{\dt \dt
    \omega}{p+1}+\frac{\delt\delt\omega}{n-p+1}\right)\, ,
\ee 
using the bracket notation of appendix \ref{cliffprods}. Expanding the commutators gives
\be
A=\frac{1}{p+1}\Bigl[d(\dt \omega)-\frac{1}{6}T\cwedge{3}\dt \omega\Bigr]+\frac{\delta(\delt \omega)-T\wedge\delt\omega}{n-p+1}-\frac{1}{2}dT\cwedge{1}\omega+\frac{1}{12}dT\cwedge{3}\omega\,. \label{anomaly}
\ee
The anomaly $A$ splits as $A^{(cl)}$ and $A^{(q)}$ a $(p+2)$-form and a $(p-2)$-form, respectively, so that $A=A^{(cl)}+A^{(q)}$. In order that $A$ vanishes, both parts must vanish separately. We have:
\ba
A^{(cl)}\!\!&=&\!\!\frac{d(\dt \omega)}{p+1}-\frac{T\wedge \delt \omega}{n-p+1}-\frac{1}{2}
dT\cwedge{1}\omega\,,\\
A^{(q)}\!\!&=&\!\!\frac{\delta(\delt \omega)}{n-p+1}-\frac{1}{6(p+1)}T\cwedge{3}\dt \omega+\frac{1}{12}dT\cwedge{3}\omega\,.
\ea
We shall sometimes refer to $A^{(cl)}$ as the `classical' anomaly, as
this anomaly shows up even in the semi-classical spinning particle
limit.

For convenience, we state some alternative forms of the operator $L_\omega$. Re-writing in terms of the Levi-Civita connection, 
$\nabla_{X_a}=\nt{a}+\frac{1}{4}X_a\hook T$, we find
\ba\label{Ldef2}
L_\omega\!\!&=&\!\!(X^a\hook \omega+e^a\wedge \omega)\nabla_{X_a}+\frac{p}{p+1} d\omega -\frac{n-p}{n-p+1}\delta\omega\nonumber\\
&&-\frac{1}{4}T\wedge \omega +\frac{3-p}{4(p+1)}T\cwedge{1}\omega+\frac{n-p-3}{8(n-p+1)}T\cwedge{2}\omega+\frac{1}{24}T\cwedge{3}\omega\,.
\ea 
Using further the correspondence \eq{gammacor}, one can rewrite this operator in `gamma-matrix notations'. Defining 
$\tilde L_\omega=(p-1)!L_\omega$, we find 
\ba
\tilde L_\omega\!\!&=&\!\!\Bigl[\omega^a_{\ b_1\dots b_{p-1}}\gamma^{b_1\dots b_{p-1}}+\frac{1}{p(p+1)}\omega_{b_1\dots b_p}\gamma^{ab_1\dots b_p}\Bigr]\nabla_a
+\frac{1}{(p+1)^2}(d\omega)_{b_1\dots b_{p+1}}\gamma^{b_1\dots b_{p+1}}\nonumber\\
&&-\frac{n-p}{n-p+1}(\delta\omega)_{b_1\dots b_{p-1}}\gamma^{b_1\dots b_{p-1}}
-\frac{1}{24}T_{b_1 b_2 b_3} \omega_{b_4\dots b_{p+3}}\gamma^{b_1\dots b_{p+3}}
+\frac{3-p}{8(p+1)}T^a_{\ b_1b_2}\omega_{a b_3\dots b_{p+1}}\gamma^{b_1\dots b_{p+1}}\nonumber\\
&&+\frac{(n-p-3)(p-1)}{8(n-p+1)}T^{ab}_{\ \ b_1}\omega_{ab b_2\dots b_{p-1}}\gamma^{b_1\dots b_{p-1}}
+\frac{(p-1)(p-2)}{24} T^{abc}\omega_{abc b_1\dots b_{p-3}}\gamma^{b_1\dots b_{p-3}}\,.\qquad
\ea

\subsection{(Anti)-Commuting operators for massive Dirac}
We now deduce Corollaries \ref{GKYcor} and \ref{GCCKYcor}. Suppose now that the anomaly $A$ vanishes, then
\be
\mD L_\omega=\omega\mD^2+\frac{(-1)^p}{p+1}\dt \omega \mD +\frac{(-1)^p}{n-p+1}\delt\omega\mD\,.
\ee
Let us further define 
\ba
K_\omega\!\!&\equiv&\!\! L_\omega-(-1)^p\omega\mD=2X^a\hook \omega\nt{a}+\frac{p}{p+1}\dt \omega-\frac{n-p}{n-p+1}\delt \omega +T\cwedge{1}\omega-\frac{1}{6}T\cwedge{3}\omega\,,\qquad\\
M_\omega\!\!&\equiv&\!\! L_\omega+(-1)^p\omega\mD=2e^a\wedge \omega\nt{a}+\frac{p}{p+1}\dt \omega-\frac{n-p}{n-p+1}\delt \omega +T\wedge\omega-\frac{1}{2}T\cwedge{2}\omega\,.
\ea
Then, by using 
\ba
\mD\omega\!\!&=&\!\!\Bigl(e^a\nt{a}+\frac{1}{2}T\Bigr)\omega=e^a\omega+\dt \omega-\delt \omega+\frac{1}{2}T\omega\nonumber\\
\!\!&=&\!\!(-1)^p\omega\mD+2X^a\hook\omega\nt{a}+\dt \omega-\delt\omega+\frac{1}{2}[T,\omega]_{p+1}\nonumber\\
\!\!&=&\!\!-(-1)^p\omega\mD+2e^a\wedge\omega\nt{a}+\dt \omega-\delt\omega+\frac{1}{2}[T,\omega]_{p}\,,\nonumber
\ea
we find that these new operators satisfy
\be
[\mD, K_\omega]_{p}=\frac{2(-1)^p}{n-p+1}\delt\omega\mD\,,\qquad 
[\mD, M_\omega]_{p+1}=\frac{2(-1)^p}{p+1}\dt\omega\mD\,.
\ee
We note that more generally, if the anomaly $A$ does not vanish, then it will appear as the right-hand-side of both of these commutators.
Let us now consider the cases when $\omega$ is a GKY and GCCKY separately. 

\subsubsection{GKY and corresponding operator}
When $\omega$ is a GKY tensor ($\delt\omega=0$), the conditions for vanishing of the anomalies $A^{(cl)}$ and $A^{(q)}$ reduce to 
\be\label{an1}
d(T\cwedge{1}\omega)+\frac{p+1}{2}\,dT\cwedge{1}\omega=0\,,\qquad 
dT\cwedge{3}\omega-\frac{2}{p+1}T\cwedge{3}\dt\omega=0\,,
\ee
and the operator $K_\omega$, obeying $[\mD, K_\omega]_{p}=0$ becomes 
\be
K_\omega=2X^a\hook \omega\nt{a}+\frac{p}{p+1}\dt \omega +T\cwedge{1}\omega-\frac{1}{6}T\cwedge{3}\omega\,.
\ee
The first condition \eq{an1} is present already at the spinning particle level 
(see, e.g., \cite{DeJonghe:1996fb} where it is derived for a GKY 2-form).
It may be useful to rewrite $K_\omega$ in terms of the Levi-Civita connection. A calculation in the Clifford algebra gives
\be
K_\omega=2X^a\hook \omega\nabla_{X_a}+\frac{p}{p+1}d\omega +\frac{1-p}{2(p+1)}T\cwedge{1}\omega-\frac{1}{2}T\cwedge{2}\omega+\frac{1}{12}T\cwedge{3}\omega\,.
\ee
Introducing $\tilde K_\omega\equiv K_\omega (p-1)!/2$, we find 
\ba
\tilde K_\omega&=&\omega^a_{\ \,b_1\dots b_{p-1}}\gamma^{b_1\dots b_{p-1}}\nabla_a+\frac{1}{2(p+1)^2}(d\omega)_{b_1\dots b_{p+1}}\gamma^{b_1\dots b_{p+1}}\nonumber\\
&&+ \frac{1-p}{8(p+1)}T^a_{\ \,b_1b_2}\omega_{ab_3\dots b_{p+1}}\gamma^{b_1\dots b_{p+1}}
-\frac{p-1}{4}T^{ab}_{\ \ \,b_1}\omega_{ab b_2\dots b_{p-1}}\gamma^{b_1\dots b_{p-1}}\nonumber\\
&&+ \frac{(p-1)(p-2)}{24}T^{abc}\omega_{abc b_1\dots b_{p-3}}\gamma^{b_1\dots b_{p-3}}\,.
\ea 
The first two terms correspond to symmetry operator discussed in \cite{Benn:1996ia, Cariglia:2003kf} in the absence of torsion.
The third term is a `leading' torsion correction, present already at the spinning particle level \cite{Kubiznak:2009qi}. The last two terms are `quantum corrections' due to the presence of torsion.

\subsubsection{GCCKY and corresponding operator}

When $\omega$ is a GCCKY tensor ($\dt\omega=0$), the conditions for vanishing of the anomalies $A^{(cl)}$ and $A^{(q)}$ reduce to  
\be\label{an2}
A^{(cl)}=-\frac{T\wedge\delt\omega}{n-p+1}-\frac{1}{2}dT\cwedge{1}\omega=0\,,\qquad 
A^{(q)}=\frac{-1}{2(n-p+1)}\delta(T\cwedge{2}\omega)+\frac{1}{12}dT\cwedge{3}\omega=0\,,
\ee
and the operator $M_\omega$, obeying $[\mD, M_\omega]_{p+1}=0$ becomes 
\be
M_\omega=2e^a\wedge \omega\nt{a}-\frac{n-p}{n-p+1}\delt \omega +T\wedge\omega-\frac{1}{2}T\cwedge{2}\omega\,.
\ee
Again, we may express this operator in terms of the Levi-Civita connection as
\be
M_\omega=2e^a\wedge \omega\nabla_{X_a}-\frac{n-p}{n-p+1}\delta\omega -\frac{1}{2}T\wedge\omega+T\cwedge{1}\omega+\frac{n-p-1}{4(n-p+1)}T\cwedge{2}\omega\,,
\ee
and introducing $\tilde M_\omega\equiv M_\omega p!/2$, we find 
\ba
\tilde M_\omega&=&\omega_{b_1\dots b_{p}}\gamma^{ab_1\dots b_{p}}\nabla_a-\frac{p(n-p)}{2(n-p+1)}(\delta\omega)_{b_1\dots b_{p-1}}\gamma^{b_1\dots b_{p-1}}\nonumber\\
&&- \frac{1}{24}T_{b_1b_2 b_3}\omega_{b_4\dots b_{p+3}}\gamma^{b_1\dots b_{p+3}}
+\frac{p}{4}T^{a}_{\ \,b_1 b_2}\omega_{ab_3\dots b_{p+1}}\gamma^{b_1\dots b_{p+1}}\nonumber\\
&&+ \frac{p(p-1)(n-p-1)}{8(n-p+1)}T^{ab}_{\ \ \,b_1}\omega_{abb_2\dots b_{p-1}}\gamma^{b_1\dots b_{p-1}}\,.
\ea 
The first two terms correspond to symmetry operator discussed by \cite{Benn:1996ia} in the absence of torsion.

We noted above that GCCKY tensors form an algebra under the wedge
product. 
Consider now $\omega=\omega_1\wedge \omega_2$, where $\omega_i$ are
GCCKY tensors. From Lemma \ref{gcckylem}
we know that \mbox{$ \delta^T \circ (n-\pi+1)^{-1}$} is a graded derivation on the
exterior algebra of GCCKY forms. Since $T$ is a $3$-form, we deduce that \mbox{$T\wedge
\delta^T\circ(n-\pi+1)^{-1}$} is a derivation. A short calculation shows
that
\be
dT\cwedge{1}(\omega_1\wedge \omega_2) = (dT\cwedge{1}\omega_1)\wedge
\omega_2+\omega_1 \wedge (dT\cwedge{1}\omega_2)
\ee
so that
\be
A^{(cl)}{(\omega_1 \wedge \omega_2)} = A^{(cl)}{(\omega_1)} \wedge \omega_2 +
\omega_1 \wedge A^{(cl)}{(\omega_2)}\,.
\ee
Hence we have derived the following lemma:
\begin{lemma}
Let $\omega_1$, $\omega_2$ are GCCKY tensors for which $A^{(cl)}{(\omega_i)}=0$. Then the classical anomaly vanishes also for a GCCKY tensor $\omega=\omega_1\wedge \omega_2$, $A^{(cl)}{(\omega_1\wedge\omega_2)}=0.$ 
\end{lemma}

\subsection{Examples}

In \cite{Houri:2010} we show that the charged Kerr-NUT spacetimes of
Chow \cite{Chow:2008fe} admit the towers of GCCKY and GKY forms for which the
anomaly vanish. We additionally demonstrate the separability of the
Dirac equation directly. 

Another, more general, situation where the anomaly simplifies
dramatically is the case of a K\"ahler with torsion (KT) or hyper-K\"ahler
with torsion (HKT) manifold \cite{Howe:1996kj,
  Grantcharov:1999kv}. Such manifolds have a range of applications in physics from
supersymmetric sigma models with Wess-Zumino term \cite{Gates:1984nk} to the construction of solutions in five dimensional de
Sitter supergravity \cite{Grover:2008jr}.

A Hermitian manifold possesses an integrable complex
structure $J$ and a metric $g$ compatible with $J$ so that $g(JX,
JY)=g(X,Y)$. The K\"ahler form $F$ is defined to be $F(X,Y) = g(JX,
Y)$. A result of Gauduchon \cite{Gauduchon:1997} shows that there
exists a unique connection $\nabla^T$ with skew-symmetric torsion such
that
\be
\nabla^T g = 0, \qquad \nabla^T J = 0.
\ee
A Hermitian manifold with this choice of connection is called a KT
manifold. The holonomy of such a connection lies in $U(n)$. If the
torsion vanishes, then the manifold is in fact K\"ahler.

The K\"ahler form is clearly preserved by the connection $\nabla^T$
and is therefore trivially a GCKY form. Since $d^T F = \delta^T F =
0$, the anomaly simplifies to give
\be
A=\frac{1}{4}\,[ dT,F]\,.
\ee 
Which may or may not vanish. A KT-manifold is said to be \emph{strong} if $T$ is closed. In this
case we find the anomaly vanishes and $F$ generates a symmetry of the
Dirac operator. In fact, since $F$ is both GKY and GCCKY we
have a symmetry of the massive Dirac equation.

A hyper-Hermitian manifold possesses three integrable complex
structures $I_1$, $I_2$, $I_3$ satisfying
\be
I_i I_j = -\delta_{ij}+\epsilon_{ijk}I_k
\ee
and a metric $g$ compatible with the $I_i$ so that
\be
g(X,Y)=g(I_1X, I_1 Y)=g(I_2X, I_2 Y)=g(I_3X, I_3 Y).
\ee
If there exists a connection with skew-symmetric torsion $\nabla^T$
such that 
\be
\nabla^T g = 0, \qquad \nabla^T I_1 =\nabla^T I_2 = \nabla^T I_3 =  0,
\ee
then $\nabla^T$ is a HKT-connection and $M$ is a HKT manifold. Note
that by Gauduchon's result, if $\nabla^T$ exists it is unique. The
holonomy of the connection in this case lies in $Sp(n)$. We may define the triplet of K\"ahler forms by $F_i(X,Y)=g(I_iX, Y)$. Each
K\"ahler form is again a GCKY, with anomaly 
\be
A_i=\frac{1}{4}\,[ dT,F_i]\,.
\ee 
As in the KT case, a \emph{strong} HKT structure has a closed $T$ and
therefore the K\"ahler forms generate symmetries of the Dirac operator.

\section*{Acknowledgments}
We wish to thank G.~W.~Gibbons for useful discussions. D.K. acknowledges the Herschel Smith Postdoctoral Research Fellowship
at the University of Cambridge. The work of Y.Y. is supported by the Grant-in Aid for Scientific Research
No.21244003 from Japan Ministry of Education. He is also grateful for the hospitality of
DAMTP, University of Cambridge during his stay.

\appendix

\section{Appendix}

\subsection{Clifford Algebra Relations \label{cliffprods}}
In this appendix we gather various identities for the Clifford product used in the main text.
We define the brackets as follows:
\begin{equation}
[ \alpha, \beta ]_p = \alpha \beta + (-1)^{p} \beta \alpha.
\end{equation}
and will make use of the shorthand $[ \alpha, \beta ]_\pm = \alpha
\beta \pm \beta \alpha$. 

Let $\alpha \in \Omega^q(M)$ and $\beta \in \Omega^p(M)$, $q\leq p$.
Then, the following formulae can be derived from
the defining Clifford algebra relations (\ref{Cliff}):\footnote{%
The derivation of these relations goes as follows. Let us calculate a coefficient standing by the term with $n$ contraction, $(\alpha\beta)_n$. We have
\ba
(\alpha\beta)_n\!\!&=&\!\!\alpha_{a_1\dots a_q}(e^{a_1\dots a_q}\beta)_n=
(-1)^{n(q-n)}\alpha_{a_1\dots a_n a_{n+1}\dots a_q}
(e^{a_{n+1}\dots a_qa_1\dots a_n}\beta)_n\nonumber\\
\!\!&=&\!\!\frac{(-1)^{n(q-n)}q!}{(q-n)!n!}\alpha_{a_1\dots a_n a_{n+1}\dots a_q}
e^{a_{n+1}\dots a_q}\wedge X^{a_1}\hook X^{a_2}\dots X^{a_n}\hook \beta\,.\nonumber
\ea
Using now the fact that 
\be
\alpha_{a_1\dots a_n a_{n+1}\dots a_q} e^{a_{n+1}\dots a_q}
=\frac{(q-n)!}{q!}X^{a_n}\hook X^{a_{n-1}}\dots X^{a_1}\hook \alpha=
(-1)^{[n/2]}\frac{(q-n)!}{q!}X^{a_1}\hook X^{a_2}\dots X^{a_n}\hook \alpha\,,\nonumber
\ee
we arrive at the first formula \eq{formulas}. The derivation of the second formula is analogous.
}
\be\label{formulas}
\alpha \beta=\sum_{n=0}^{q}\frac{(-1)^{n(q-n)+[n/2]}}{n!}\alpha\cwedge{n}\beta\,,\quad
\beta\alpha=(-1)^{pq}\sum_{n=0}^{q}\frac{(-1)^{n(q-n+1)+[n/2]}}{n!}\alpha\cwedge{n}\beta\,.
\ee
The corresponding brackets can be easily deduced.
In particular, in the main text we shall use the following relations for $q\leq 4$:
 
Let $\alpha \in \Omega^1(M), \beta \in \Omega^p(M)$, then
\begin{eqnarray}
\textstyle
 \alpha \beta &=& \alpha \cwedge{0} \beta + \alpha \cwedge{1} \beta, \\
 \beta \alpha &=& (-1)^p(\alpha \cwedge{0} \beta - \alpha \cwedge{1}
 \beta), \\
 {[} \alpha, \beta {]}_p &=& 2 \alpha \cwedge{0} \beta,\\
 {[} \alpha, \beta {]}_{p+1} &=& 2 \alpha \cwedge{1} \beta.
\end{eqnarray}

Let $\alpha \in \Omega^2(M), \beta \in \Omega^p(M)$, then
\begin{eqnarray}
 \alpha \beta &=& \alpha \cwedge{0} \beta - \alpha \cwedge{1} \beta -
 {\textstyle{\frac{1}{2}}} \alpha \cwedge{2} \beta, \\
 \beta \alpha &=& \alpha \cwedge{0} \beta + \alpha \cwedge{1}
 \beta -{\textstyle{\frac{1}{2}}} \alpha \cwedge{2} \beta,\\
 {[} \alpha, \beta {]}_+ &=& 2 \alpha \cwedge{0}\beta- \alpha \cwedge{2}  \beta, \\
 {[} \alpha, \beta {]}_- &=& -2 \alpha \cwedge{1} \beta.
\end{eqnarray}

Let $\alpha \in \Omega^3(M), \beta \in \Omega^p(M)$, then
\begin{eqnarray}
 \alpha \beta &=& \alpha \cwedge{0} \beta + \alpha \cwedge{1} \beta -
 {\textstyle{\frac{1}{2}}} \alpha \cwedge{2} \beta - {\textstyle{\frac{1}{6}}}\alpha \cwedge{3} \beta, \\
 \beta \alpha &=& (-1)^p( \alpha \cwedge{0} \beta - \alpha \cwedge{1}
 \beta -{\textstyle{\frac{1}{2}}} \alpha \cwedge{2} \beta+{\textstyle{\frac{1}{6}}}\alpha \cwedge{3} \beta),\\
 {[} \alpha, \beta {]}_p &=& 2 \alpha \cwedge{0}\beta- \alpha \cwedge{2}  \beta, \\
 {[} \alpha, \beta {]}_{p+1} &=& 2 \alpha \cwedge{1} \beta -
 {\textstyle{\frac{1}{3}} }\alpha \cwedge{3} \beta.
\end{eqnarray}

Let $\alpha \in \Omega^4(M), \beta \in \Omega^p(M)$, then
\begin{eqnarray}
 \alpha \beta &=& \alpha \cwedge{0} \beta - \alpha \cwedge{1} \beta -
 {\textstyle{\frac{1}{2}}} \alpha \cwedge{2} \beta +
 {\textstyle{\frac{1}{6}}}\alpha \cwedge{3} \beta
 +{\textstyle{\frac{1}{24}}} \alpha \cwedge{4} \beta,\\
 \beta \alpha &=&  \alpha \cwedge{0} \beta + \alpha \cwedge{1}
 \beta -{\textstyle{\frac{1}{2}} }\alpha \cwedge{2} \beta-{\textstyle{\frac{1}{6}}}\alpha \cwedge{3} \beta+{\textstyle{\frac{1}{24}}} \alpha \cwedge{4} \beta,\\
 {[} \alpha, \beta {]}_+ &=& 2 \alpha \cwedge{0}\beta- \alpha \cwedge{2}
 \beta + {\textstyle{\frac{1}{24}} }\alpha \cwedge{4} \beta, \\
 {[} \alpha, \beta {]}_- &=& 2 \alpha \cwedge{1} \beta -
 {\textstyle{\frac{1}{3}}} \alpha \cwedge{3} \beta.
\end{eqnarray}

\subsection{Contracted Wedge Product Identities}

In order to work with the contracted wedge product introduced
in Section \ref{Cliffsec}, it is convenient to derive some identities to help with
calculations. We collect some results here. Firstly some algebraic results.

\begin{lemma} \label{cwplem1}
Suppose $\alpha \in \Omega^p(M)$, then the interior derivative $X \hook$ satisfies
\be
X \hook (\alpha \cwedge{n}\beta) = (-1)^{n} (X \hook \alpha)
\cwedge{n} \beta + (-1)^p \alpha \cwedge{n} (X \hook \beta)\,.
\ee
\end{lemma}
\begin{proof}
By induction. The $n=0$ case is trivial. Suppose true for $n-1$, we have
\begin{eqnarray*}
X \hook (\alpha \cwedge{n}\beta) &=& X \hook (X_a \hook \alpha
\cwedge{n-1} X^a \hook \beta) \nonumber \\ &=& (-1)^{n-1}(X \hook X_a
\hook \alpha \cwedge{n-1} X^a \hook \beta)+ (-1)^{p-1}( X_a
\hook \alpha \cwedge{n-1} X \hook X^a \hook \beta) \nonumber \\ &=& (-1)^{n} (X \hook \alpha)
\cwedge{n} \beta + (-1)^p \alpha \cwedge{n} (X \hook \beta)\,.
\end{eqnarray*}
\end{proof}

\begin{lemma} \label{cwplem2}
Suppose $\omega \in \Omega^1(M)$ and  $\alpha \in \Omega^p(M)$, then
\begin{eqnarray}
(\omega \wedge \alpha) \cwedge{n} \beta &=& (-1)^n \omega \wedge
(\alpha \cwedge{n} \beta) + n\, \alpha \cwedge{n-1} (\omega^\sharp
\hook \beta)\,, \label{cwpid1}\\
\alpha \cwedge{n} (\omega \wedge \beta) &=& (-1)^p \omega \wedge
(\alpha \cwedge{n} \beta) + n \, (\omega^\sharp
\hook \alpha ) \cwedge{n-1} \beta\,. \label{cwpid2}
\end{eqnarray}
\end{lemma}
\begin{proof}
Consider (\ref{cwpid1}). By linearity it suffices to consider the case
$\omega = e^a$. The $n=0$ case is straightforward. Suppose true for
$n-1$. Then
\begin{eqnarray*}
(e^a \wedge \alpha) \cwedge{n} \beta &=& X_b \hook (e^a \wedge \alpha)
\cwedge{n-1} X^b \hook \beta = \alpha \cwedge{n-1} X_a \hook \beta -
e^a\wedge (X_b \hook \alpha) \cwedge{n-1} X^b \hook \beta \nonumber \\
&=& \alpha \cwedge{n-1} X_a \hook \beta - (-1)^{n-1} e^a \wedge( X_b
\hook \alpha \cwedge{n-1} X^b \hook \beta)-(n-1)  X_b
\hook \alpha \cwedge{n-1}X^a \hook  X^b \hook \beta \nonumber \\
&=& (-1)^n e^a \wedge
(\alpha \cwedge{n} \beta) + n\, \alpha \cwedge{n-1} (X^a
\hook \beta).
\end{eqnarray*}
An entirely analogous calculation establishes (\ref{cwpid2}).
\end{proof}

We now consider the action of $\nt{}$, $\delta^T$ and $d^T$ on the
contracted wedge product

\begin{lemma}\label{cwplem3}
$\nt{}$ is a derivation over the contracted wedge product:
\be\label{A22}
\nt{} (\alpha \cwedge{n}\beta) = \nt{} \alpha \cwedge{n}\beta+
\alpha \cwedge{n}\nt{} \beta\,.
\ee
Moreover, when $\alpha \in \Omega^p(M)$, then
\begin{eqnarray}
\delta^T(\alpha \cwedge{n}\beta) &=& (-1)^n \delta^T \alpha \cwedge{n}
\beta - (-1)^p \nt{a} \alpha \cwedge{n} X^a \hook \beta \nonumber \\
&& \quad +(-1)^p  \alpha \cwedge{n}
\delta^T \beta - (-1)^n X^a \hook \alpha \cwedge{n} \nt{a} \beta\,.\label{A23}\\
d^T(\alpha \cwedge{n}\beta) &=& (-1)^n (d^T \alpha \cwedge{n}
\beta - n \nt{a} \alpha \cwedge{n-1} X^a \hook \beta )\nonumber \\
&& \quad +(-1)^p ( \alpha \cwedge{n}
d^T \beta - n X^a \hook \alpha \cwedge{n} \nt{a} \beta )\,.\label{A24}
\end{eqnarray}
\end{lemma}
\begin{proof}
Let us first prove \eq{A22}. This is clearly true for $n=0$. Suppose true for $n-1$, then
\begin{eqnarray*}
\nt{} (\alpha \cwedge{n}\beta) &=& \nt{} (X_a \hook \alpha
\cwedge{n-1} X^a \hook \beta) \nonumber \\ &=& \nt{} (X_a \hook \alpha)
\cwedge{n-1} X^a \hook \beta + X_a \hook \alpha 
\cwedge{n-1} \nt{} (X^a \hook \beta)\nonumber \\
&=& \frac{1}{2} T(X_b, X^a, X_c) X^c \hook \alpha \cwedge{n-1} X_a
\hook \beta + X_a \hook \nt{} \alpha \cwedge{n-1}X^a \hook \beta
\nonumber \\ && \quad + X_a \hook \alpha \cwedge{n-1}\frac{1}{2} T(X_b, X_a,
X^c) X_c \hook \beta + X_a \hook \alpha \cwedge{n-1} X^a \hook \nt{}
\beta \nonumber \\ &=& \nt{} \alpha \cwedge{n}\beta+
\alpha \cwedge{n}\nt{} \beta\,,
\end{eqnarray*}
using the anti-symmetry of the torsion. The properties \eq{A23} and \eq{A24} follow 
from \eq{A22}, Lemmas \ref{cwplem1}, \ref{cwplem2} and the fact that  $\delta^T = - X^a \hook \nt{a}$, $\delta^T = e^a \wedge \nt{a}$.
\end{proof}

\subsection{Bianchi Identities \label{bianchisec}}

We derive some Bianchi identities by repeated
application of the $T$-external derivative (\ref{extder}) to a general
form $\omega$. Working in a basis satisfying (\ref{basiscond})
\begin{eqnarray}
d^T d^T \omega &=& e^a \wedge \nabla^T_{X_a} \left(e^b \wedge
  \nabla^T_{X_b} \right) \nonumber\\ &=& e^a \wedge \frac{1}{2}T(X_a X^b, X_c)e^c
\wedge \nabla^T_{X_b} \omega + \frac{1}{2} e^a \wedge e^b \wedge
\left[\nabla^T_{X_a}, \nabla^T_{X_b} \right ] \omega \nonumber \\ &=& -X^a \hook
T \wedge \nabla^T_{X_a} \omega + \frac{1}{2} e^a \wedge e^b \wedge
R(X_a, X_b) \omega\,.
\end{eqnarray}
Alternatively, we may express the $T$-exterior derivative as
\begin{eqnarray}
d^T d^T\omega &=& (d-T \cwedge{1} ) (d\omega-T \cwedge{1}
  \omega) \nonumber \\
&=& - T \cwedge{1} d^T \omega - d^T ( T \cwedge{1} \omega ) - T
{\cwedge{1}} (T \cwedge{1} \omega)\,.
\end{eqnarray}
Now, making use of Lemma \ref{cwplem3}, we find
\begin{equation}
\frac{1}{2} e^a \wedge e^b \wedge R(X_a, X_b) \omega = d^TT \cwedge{1}
\omega - \nabla^T_{X_b} T \wedge X^b \hook \omega -  T
{\cwedge{1}} (T \cwedge{1} \omega).
\end{equation}

By an analogous calculation we deduce that
\begin{equation}
\delta^T \delta^T \omega = \frac{1}{2} X^a \hook
T \cwedge{2} \nabla^T_{X_a} \omega + \frac{1}{2} X^a \hook X^b \hook
R(X_a, X_b) \omega \label{bianchi1}
\end{equation}
and
\begin{equation}
X^a \hook X^b \hook R(X_a, X_b) \omega = -\delta^TT \cwedge{2}
\omega - \nabla^T_{X_b} T \cwedge{2} X^b \hook \omega -  \frac{1}{2}T
{\cwedge{2}} (T \cwedge{2} \omega).
\end{equation}

We now deduce a Bianchi identity for the
curvature forms $R_{ab}$ defined by 
\[
R_{ab}(X,Y) = X_a \hook R(X,Y)
e_b\,.
\]
These are anti-symmetric on their indices as may be shown by
considering the curvature operator acting on $g^{-1}(e_a, e_b)$. We
calculate
\begin{equation}
R_{ab} \wedge e^b = -R_{ba}\wedge e^b = -\frac{1}{2}(X_b \hook R(X_c,
X_d)e_a) e^b \wedge e^c \wedge e^d = -\frac{1}{2} e^c \wedge e^d
\wedge R(X_c,X_d)e_a.
\end{equation}
Making use of (\ref{bianchi1}) we find
\begin{equation}
R_{ab} \wedge e^b = \nabla^T_{X_a} T - X_a \hook d^T T + (X_a \hook T)
\cwedge{1} T
\end{equation}
and readily deduce
\begin{eqnarray}
R_{ab} \wedge e^a \wedge e^b &=& -3 d^T T - 2 T \cwedge{1} T\,, \\
X^a \hook (R_{ab} \wedge e^b) &=& -\delta^T T.
\end{eqnarray}
Now, recall the definition of $\mathcal{R}_{X,Y}$
\begin{equation}
\mathcal{R}_{X,Y} = -\frac{1}{4} X \hook Y \hook R_{ab} e^{ab}.
\end{equation}
Making use of the rules for Clifford products we find
\begin{eqnarray}
e^{ab} \mathcal{R}_{X_a, X_b} &=& -\frac{1}{4}e^{ab} X_a \hook Y_b \hook
R_{cd} e^{cd} = \frac{1}{2}R_{ab} e^{ab} \nonumber \\
&=& \frac{1}{2}R_{ab}\wedge e^a \wedge e^b - X_a\hook(e^b \wedge
R_{ab}) + \frac{1}{2}X^a\hook X^b \hook R_{ab} \nonumber \\
&=& -\frac{3}{2} d^TT - T \cwedge{1} T + \delta^T T - \frac{1}{2} s.
\end{eqnarray}
Here $s$ is the scalar curvature of the connection with torsion, which
we take to be defined by
\begin{equation}
s = -X^a \hook X^b \hook R_{ab} = - X^a \hook R(X_a, X_b) e^b.
\end{equation}

\end{document}